\documentclass[twocolumn,superscriptaddress,nofootinbib,preprintnumbers,showpacs,tightenlines,notitlepage]{revtex4}
\usepackage[latin1]{inputenc}
\usepackage{graphicx}
\usepackage{amssymb}
\usepackage{color}
\usepackage{float} 
\usepackage{amsmath}
\usepackage{amsfonts}  
\usepackage{dcolumn}
\usepackage{hyperref} 
\usepackage{amsthm}
\usepackage{color}

\def\3nab{\tilde{\nabla}}

\def\la {\langle}
\def\ra {\rangle}

\def\be {\begin{equation}}
	\def\ee {\end{equation}} 
\def\ba {\begin{eqnarray}}
	\def\ea {\end{eqnarray}}

\newtheorem{prop}{Proposition}
\newtheorem{thm}{Theorem}
\newtheorem{cor}{Corollary}

\newcommand{\barray}{\begin{array}}
	\newcommand{\earray}{\end{array}}

 \newcommand{\nab}{\nabla}

\newcommand \ep {\epsilon}

\newcommand \om {\omega}

\begin{document}
	
	\title{Necessity of spacetime shear for cosmological gravitational waves}

	\author{Roger M. Mayala}
	\email{rmayala2@gmail.com }
	\affiliation{Astrophysics Research Centre, School of Mathematics, Statistics and Computer Science, University of KwaZulu-Natal, Private Bag X54001, Durban 4000, South Africa.}
	\author{Rituparno Goswami}
	\email{Goswami@ukzn.ac.za}
	\affiliation{Astrophysics Research Centre, School of Mathematics, Statistics and Computer Science, University of KwaZulu-Natal, Private Bag X54001, Durban 4000, South Africa.}
	\author{Sunil D. Maharaj}
	\email{Maharaj@ukzn.ac.za}
	\affiliation{Astrophysics Research Centre, School of Mathematics, Statistics and Computer Science, University of KwaZulu-Natal, Private Bag X54001, Durban 4000, South Africa.}

	\begin{abstract}
         We show that a general but shear-free perturbation of homogeneous and isotropic universes are necessarily silent, without any gravitational waves. We prove this in two steps. First we establish that a shear free perturbation of these universes are acceleration-free and the fluid flow geodesics of the background universe maps onto themselves in the perturbed universe. This effect then decouples the evolution equations of the electric and magnetic parts of the Weyl tensor in the perturbed spacetimes and the magnetic part no longer contains any tensor modes. Although the electric part, that drives the tidal forces, do have tensor modes sourced by the anisotropic stress, these modes have homogeneous oscillations at every point on a time slice without any wave propagation. We also show the presence of vorticity vector waves that are sourced by the curl of heat flux. This analysis shows the critical role of the shear tensor in generating cosmological gravitational waves in an expanding universe.
	\end{abstract}
	

	\maketitle
	\section{Introduction}
	
This paper analyses the role of spacetime  shear in generating cosmological gravitational waves in a universe constructed via a general perturbation of the homogeneous and isotropic Friedmann Lemaitre Robertson Walker (FLRW) universe. This is important, as the corresponding Newtonian, or quasi Newtonian description of cosmologies, do not permit the existence of gravitational waves. It has been well known for some time that we can describe the cosmology in a quasi-Newtonian way (the {\it Silent models}, that are devoid of any gravitational waves), for observers which move along geodesics which are both shear-free and irrotational \cite{ Matarrese}. The key difference between Newtonian and general relativistic cosmologies, in the absence of shear, emerges from the surprising exact result: A shear-free dust cannot rotate and expand simultaneously \cite{Godel, Ellis, Ellisbook}, which was also shown to hold in the case of barotropic perfect fluid solutions linearised about a FLRW geometry \cite{Nzioki}. Since this result is not true for Newtonian or quasi Newtonian cosmologies, it is not always obvious which behaviour of a perturbed universe (linearised about FLRW geometry) will have a Newtonian counterpart. Further to 
this, most of the previous results have assumed the form of matter in the perturbed universe to be barotropic (as in the background scenario). Hence the effect of introducing heat flux or anisotropic stress perturbatively in the energy momentum tensor of the matter field still remains the subject of further investigations. \\

To discuss the generation of gravitational waves on a given background, the covariant way is to consider the Weyl tensor as the free gravitational field and the metric tensor as it's second order potential field. Given a family of timelike observers, one can then split the Weyl tensor into its electric and magnetic parts, which are projected symmetric trace-free tensors of rank 2, and this is absolutely analogous to splitting the electromagnetic field tensor into the electric field vector and magnetic field vector. The once-contracted Bianchi identities then give the evolution equations of these tensors and one can easily see that the evolution of the electric part is coupled to the curl of the magnetic part and vice versa, exactly analogous to the source-free Maxwell equations for electric and magnetic fields. Using these evolution equations and the tensor identities, we can then find closed wave equations for the electric and magnetic part of the Weyl tensor (see \cite{GosEllis} for detailed explanation). These are the equations for gravitational (tensor) waves containing the tensorial modes of oscillations. Apparently, it may seem that the spacetime shear plays no direct role in generating these waves. However, as we shall see in this paper, the shear has a pivotal role in coupling the evolution equations for electric and magnetic Weyl. Absence of shear completely destroys this structure and makes the perturbed spacetime silent. We note that the silent universes have many interesting physical features. which have been well studied in several treatments \cite{Lesame,Roy1,Bergh,Bergh1,Bol,Bol1} \\

In this paper we work in a semi tetrad covariant and gauge invariant formalism to obtain frame invariant and gauge invariant results. As we know, in perturbation theory the gauge choice becomes important while mapping the zeroth order quantities from the background manifold to the perturbed one. However, it is also well known that any covariantly defined geometric and thermodynamic quantity, that vanishes in the background, is automatically gauge invariant in the perturbed manifold \cite{SW}. Since in the FLRW background the Weyl tensor is identically zero, in the perturbed manifold this will be a first order quantity and naturally gauge invariant. Therefore any result concerning the existence or non-existence of gravitational waves will naturally be gauge invariant in this formalism. 
Unless otherwise specified, we use natural units ($c=8\pi G=1$) throughout this paper, Latin indices run from 0 to 3. We use the $(-,+,+,+)$ signature and the Riemann tensor is defined by
\begin{equation}
R^{a}{}_{bcd}=\Gamma^a{}_{bd,c}-\Gamma^a{}_{bc,d}+ \Gamma^e{}_{bd}\Gamma^a{}_{ce}-\Gamma^e{}_{bc}\Gamma^a{}_{de}\;,
\end{equation}
and the Ricci tensor is obtained by contracting the {\em first} and the {\em third} indices of the above. The Hilbert--Einstein action in the presence of matter is given by
\begin{equation}
{\cal S}=\frac12\int d^4x \sqrt{-g}\left[R-2{\cal L}_m \right]\;,
\end{equation}
variation of which gives the Einstein's field equations as
\be
G_{ab} =T_{ab}\;. 
\ee
 	
	\subsection{Semi tetrad 1+3 splitting of spacetime manifold}
This formalism \cite{EllisCovariant,EllisVarenna,Ellisbook} is based on a local 1+3 threading of the spacetime manifold ( on any open set $\mathcal{U}$ of the manifold), and has proved itself to be an excellent tool for understanding many geometrical and physical aspects of relativistic fluid flows, both in general relativity and in the gauge invariant and covariant cosmological perturbations \cite{Ellisbook,EB}. In this formalism we first define a timelike congruence with a unit tangent vector $u^a$ ($u^au_a=-1$). The most natural choice of defining such a congruence will be along the fluid flow lines, where the vector $u^a$ defines the 4-velocity of a fluid particle. 
Once such a congruence is defined, we then split the spacetime locally in the form $R\otimes V$ where  $R$ denotes the worldline of the fluid particle and $V$ is the 3-space perpendicular to the worldine at any given point. Then we see that any 4-vector $X^a$  can then be projected on the 3-space by the projection tensor $h^a{}_b=g^a{}_b+u^au_b$.  The choice of $u^a$ naturally defines two derivatives: the \textit{covariant
time derivative} along the observers' worldlines (denoted by a dot), and the fully orthogonally \textit{projected covariant derivative} $D$ on the three dimensional space at any given instant of the observer. In other words, for any tensor $ S^{a..b}{}_{c..d}$, we have
	\be
	\dot{S}^{a..b}{}_{c..d}{} = u^{e} \nab_{e} {S}^{a..b}{}_{c..d} \;,
	\ee
	and 
	\be D_{e}S^{a..b}{}_{c..d}{} = h^a{}_f
	h^p{}_c...h^b{}_g h^q{}_d h^r{}_e \nab_{r} {S}^{f..g}{}_{p..q}\;,
	\ee 
	with total projection on all the free indices.  Angle brackets
	denote orthogonal projections of vectors, and the orthogonally
	\textit{projected symmetric trace-free} PSTF part of tensors: 
	\be
	V^{\la a \ra} = h^{a}{}_{b}V^{b}~, ~ S^{\la ab \ra} = \left[
	h^{(a}{}_c {} h^{b)}{}_d - \frac{1}{3} h^{ab}h_{cd}\right] S^{cd}\;.
	\label{PSTF} 
	\ee 
	This also naturally defines the 3-volume element 
	\be \ep_{a b c}=-\sqrt{|g|}\delta^0_{\left[ a
		\right. }\delta^1_b\delta^2_c\delta^3_{\left. d \right] }u^d\;,
	\label{eps1} 
	\ee 
	The covariant derivative of the timelike vector $u^a$ can now be
	decomposed into the irreducible part as 
	\be
	\nabla_au_b=-\dot{u}_au_b+\frac13h_{ab}\Theta+\sigma_{ab}+\ep_{a b
		c}\om^c\;, 
	\ee 
	where $\dot{u}_a$ is the acceleration,
	$\Theta=D_au^a$ is the expansion, $\sigma_{ab}=D_{\la a}u_{b \ra}$
	is the shear tensor and $\om^a=\ep^{a b c}D_bu_c$ is the vorticity
	vector. Similarly the Weyl curvature tensor can be decomposed
	irreducibly into the Gravito-Electric and Gravito-Magnetic parts as
	\ba 
	E_{ab}&=&C_{abcd}u^cu^d=E_{\la ab\ra}\;,\;\\
	H_{ab}&=&\frac12\ep_{acd}C^{cd}{}_{be}u^e=H_{\la ab\ra}\;, 
	\ea 
	which allows for a covariant description of tidal forces and gravitational
	radiation. The energy momentum tensor for a general matter field can 
	be similarly decomposed as follows:
	\be
	T_{ab}=\mu u_au_b+q_au_b+q_bu_a+ph_{ab}+\pi_{ab}\;,
	\ee
	where $\mu=T_{ab}u^au^b$ is the energy density, $p=(1/3 )h^{ab}T_{ab}$ is the isotropic pressure, $q_a=q_{\la a\ra}=-h^{c}{}_aT_{cd}u^d$ is the 3-vector defining 
	the heat flux and $\pi_{ab}=\pi_{\la ab\ra}$ is the anisotropic stress.
	
\subsection{Notations and commutations}

In terms of the fully projected spatial derivatives, we can then define the usual differential operators of vector calculus as follows:
For any projected 3-vector $V$ and second rank 3-tensor $A_{ab}$, we write
	\ba
	{\rm div}\, V=D_aV^a, & ({\rm curl}\, V)_a=\epsilon_{abc}D^bV^c, \label{id0}\\
	({\rm div}\, A)_a=D^bA_{ab}, & ({\rm curl}\, A)_{ab}=\epsilon_{cd\la a}D^cA^d_{\;b\ra}.\label{id00}
	\ea
We also note that unlike the partial derivatives, the projected covariant derivatives on the 3-space do not commute with each other, and neither do they commute with the time derivative. For example, given any scalar function $f$, we have
\ba 
D_{[a}D_{b]}{f} &=& \epsilon_{abc}\omega^c\dot{f},\label{E13}\\
\epsilon^{abc}D_bD_cf &= &2\omega^a\dot{f},\label{E14}\\
\left[D^{<a>}f\right]\dot{}\equiv   h_b^a[D^bf]\dot{}&= &D^a\dot{f}-\frac{1}{3}\Theta D^af. \label{E15}
\ea
Further to this, if $V^a$ is a projected first order 3-vector on the perturbed manifold about a FLRW background, then the linearised commutation relations are given as
\begin{align} 
\left[D^{<a}V^{b>}\right]\dot{}&= D^a\dot{V}^{<b>}-\frac{1}{3}\Theta D^aV^b,\label{E16}\\ 
({\rm div}\, V)\dot{}&= {\rm div}\, \dot{V}-\frac{1}{3}\Theta({\rm div}\, V),\label{E17}\\
D_{[a}D_{b]}V_c &= \frac{1}{3}\left(\mu - \frac{1}{3}\Theta^2\right)h_{c[a}V_{b]},\label{E18}\\
[({\rm curl}\, V)_a]\,\dot{} &=({\rm curl}\, \dot{V})_a-\frac{1}{3}\Theta({\rm curl}\, V)_a.\label{E19}  
\end{align}   
Similarly, for any first order second rank $3$-tensor $A^{ab}$, the following linearised relation holds:
\ba 
D_{[a}D_{b]}A^{cd} &=  &\frac{2}{3}\left(\mu - \frac{1}{3}\Theta^2\right)h_{\;\;[a}^{(c}A^{d)}_{\;\;b]},\label{E20}\\ 
\left[({\rm curl}\, A)_{ab}\right]\,\dot{} &=&({\rm curl}\, \dot{A})_{ab}-\frac{1}{3}\Theta({\rm curl}\, A)_{ab} .\label{E21}
\ea
Apart from these there are several other linearised relations for first order vectors and tensors in a perturbed spacetime about a FLRW background, that we list here	  \cite{Roy}  	
\be\label{id1}
	D^a({\rm curl}\, V)_a=0\,,
	\ee
	\be\label{id2}
	D^b({\rm curl}\, A)_{ab}=\frac12{\rm curl}\,(D^bA_{ab})\,,
	\ee
	\ba\label{id3}
	({\rm curl}\, {\rm curl}\,V)^a&=&D_bD^aV^b - D^2V^a \,.
	\ea
	\ba\label{id4}
	({\rm curl}\, {\rm curl}\,A)_{ab}&=&-D^2A_{ab}+\frac32D_{\la a}({\rm div}\, A)_{b\ra}\nonumber\\
	&&+\left(\mu- \frac{1}{3}\Theta^2\right)A_{ab}.
	\ea
		
In what follows, we will be using these relations repeatedly to extract the results for a shear-free perturbation of FLRW universe.

	\section{Shear-free perturbation around FLRW spacetime: Linearised field equations}

As we discussed in the introductory section, our background spacetime is a homogeneous and isotropic FLRW universe. Therefore the only non zero (zeroth order) geometric and thermodynamic quantities in the background are 
\be
\mathcal{D}_0=\{ \Theta, \mu, p\}\;.
\ee
Since the background spacetime is homogeneous and isotropic, the projected spatial derivatives of the above quantities will vanish identically, and the evolution equations that governs the temporal evolution of these quantities are given as 
\be\label{E01}
\dot{\Theta}= -\frac{1}{3}\Theta^2 -\frac{1}{2}(\mu + 3p)\;.
\ee
\be\label{E07}
\dot{\mu} = -\Theta(\mu +p)\;.
\ee
The above two equations, with a given equation of state of the form $p=p(\mu)$, will then completely solve the background system. Let us now perturb the above system by considering all the quantities that vanishes in the background be of first order of smallness in the perturbed spacetime. Along with perturbing the geometrical quantities, we also perturb the energy momentum tensor of the matter by introducing small amount of heat flux and anisotropic stress. In other words the matter in the perturbed spacetime is no longer barotropic. However, we take the shear tensor to be identically zero in the perturbed manifold. That is, we would like to show the importance of shear by the method of negation. In that case the quantities that are of first order smallness in the perturbed spacetime are given as
\be
\mathcal{D}_1=\{E_{\la ab\ra},H_{\la ab\ra},\dot{u}_{\la a\ra}, \om_{\la a\ra},q_{\la a\ra}, \pi_{\la ab\ra} \}\;.
\ee
Apart from the above set, any spatial derivative of zeroth order background quantities are also first order. As discussed before, all these first order quantities will be gauge invariant. 
The Riemann tensor of the perturbed spacetime can now be completely specified in therms of the matter variables and the Weyl variables as follows:
\begin{widetext}
\ba
R^{ab}_{\;\;\;cd}&= &2\left( 2u^{[a}u_{\;\;[c}E^{b]}_{\;\;d]}+2h^{[a}_{\;\;[c}E^{b]}_{\;\;d]}
	-u^{[a}h^{b]}_{\;\;[c}q_{d]}-u_{[c}h^{[a}_{\;\;d]}q^{b]}\right)-2\left( u^{[a}u_{[c}\pi^{b]}_{\;\;d]}-h^{[a}_{[c}\pi^{b]}_{\;\;d]}
    -\epsilon^{abe}u_{[c}H_{d]e}-\epsilon_{cde}u^{[a}H^{b]e}\right)\nonumber\\
   &&+\frac{2}{3}\left[ (\mu +3p)u^{[a}u_{[c}h^{b]}_{\;\;d]}
    +\mu h^a_{\;\;[c}h^b_{\;\;d]}\right] \;.
 \ea
 \end{widetext}
Using the above, we can now write the Ricci identities of the vector $u^a$ and once and twice contracted Bianchi identities, project them along $u^a$ and on the instantaneous spatial 3-surface of the observer and then finally linearise them by keeping the terms only to the first order smallness, to get the following set of linearised evolution equations and constraints.  
\subsection{Evolution equations}	 	 
	\be\label{E1}
	\dot{\Theta}-\mbox{div }\dot{u} =  -\frac{1}{3}\Theta^2 -\frac{1}{2}(\mu + 3p)\,,
	\ee
	\be\label{E3}
	\dot{\omega}^{< a>}-\frac{1}{2}(\mbox{curl }\dot{u})^a = -\frac{2}{3}\Theta\omega^a\,,
	\ee
	\ba
	\dot{E}^{<ab>}= & (\mbox{curl }H)^{ab}-\Theta\left(E^{ab}+\frac{1}{6}\pi^{ab}\right)\nonumber\\ 
	&-\frac{1}{2}\left(\dot{\pi}^{<ab>}+D^{<a}q^{b>}\right)\,, \label{E4}
	\ea
	\be\label{E5}
	\dot{H}^{<ab>}=-(\mbox{curl }E)^{ab} +\frac{1}{2}(\mbox{curl }\pi)^{ab} -\Theta H^{ab}\,,
	\ee
	\be\label{E6}
	\dot{q}^{\la a\ra}+D^ap +D_b \pi^{ab}=-\frac{4}{3}\Theta q^a-(\mu +p)\dot{u}^a \,,
	\ee
	\be\label{E7}
	\dot{\mu} +D_aq^a= -\Theta(\mu +p)\,.
	\ee

\subsection{Constraints}
	
	\be\label{E2}
	(C_0)^{ab} \equiv D^{<a}\dot{u}^{b>}-E^{ab}+\frac{1}{2}\pi^{ab}=0\,,
	\ee
	\be\label{E8}
	(C_1)^a \equiv q^a-\frac{2}{3}D^a\Theta+\epsilon^{abc}D_b\omega_c=0\,,
	\ee
	\be\label{E9}
	(C_2)\equiv D_a\omega^a=0 \,,
	\ee
	\be\label{E10}
	(C_3)^{ab} \equiv H^{ab}+ D^{<a}\omega^{b>}=0\,,
	\ee
	\be\label{E11}
	(C_4)^a \equiv  D_b\left( E^{ab}+\frac{1}{2}\pi^{ab}\right) -\frac{1}{3}D^a\mu +\frac{1}{3}\Theta q^a=0\,,
	\ee
	\be\label{E12}
	(C_5)^a\equiv D_bH^{ab}+(\mu +p)\omega^a +\frac{1}{2}(\mbox{curl }q)^a =0\,.
	\ee
We note here three important points:
\begin{enumerate}
\item First of all, the above system of equations is not closed. To close the system we must supply a thermodynamic relation between isotropic pressure, fluid energy density, heat flux and  anisotropic stress in the form of $F(\mu, p, q^a,\pi^{ab})=0$.
\item Secondly, the constraints $C_1, C_2,\cdots,C_5$ are the constraint equations for general matter motion, which are known to be consistently {\it time propagated} along
$u^a$ locally. However the absence of shear gives the new constraint  $C_0 $, which was the shear evolution equation in the original system. This new constraint must be {\it spatially compatible} with the original constraints and furthermore it should be consistently time propagated locally. 
\item And finally, although the spatial derivative operators $D^a$ are orthogonal to $u^a$ ( that is $u^aD_a=0$), if the vorticity is non vanishing then these operators do not span a three dimensional surface as the commutators of these spatial directional derivatives acting on a scalar do not vanish. In that case the instantaneous rest frame of an observer is not a genuine 3-surface but rather a collection of tangent planes. 
\end{enumerate}	

 \section{Spatial Consistency of the New Constraint: An important theorem}
 
In this section, we will state and prove the following important theorem on the 4-acceleration of matter for a perturbed spacetime linearised about a FLRW background.

\begin{thm}
If a homogeneous and isotropic spacetime is perturbed in a shear-free way, then the 4-acceleration of matter in the perturbed spacetime is necessarily zero, if the matter obeys the strong energy condition. In other words, the geodesics of matter flow lines in the background, map onto themselves in the perturbed manifold. 
\end{thm}
\begin{proof}
This theorem can be proved by checking the spatial consistency of the new constraint (\ref{E2}), that is, if this constraint is consistent with the existing constraints of the field equations. 
Contracting the commutation relation (\ref{E18}) we get
\begin{equation}\label{E23}
D_bD^aV^b = D^a(\mbox{div V}) + \frac{2}{3}\left(\mu - \frac{1}{3}\Theta^{2}\right)V^a,
\end{equation}
and by definition we know
\begin{equation}\label{E24} 
D^{<a}V^{b>} =\frac{1}{2}\left(D^aV^b+D^bV^a\right)-\frac{1}{3}h^{ab}(\mbox{div V}).
\end{equation}
Now, from our new constraint (\ref{E2}), we see that 
\begin{equation}\label{E25}
D_bE^{ab} = D_b(D^{<a}\dot{u}^{b>}) + \frac{1}{2}D_b \pi^{ab}.
\end{equation} 
Using relation (\ref{E24}), we get
\begin{equation}\label{E255}
  \begin{split}	
  D_bE^{ab}& = \frac{1}{2}D_bD^a\dot{u}^b + \frac{1}{2}D^2\dot{u}^a- \frac{1}{3}D^a(\mbox{div }\dot{u})\\
  & + \frac{1}{2}D_b \pi^{ab}.
  \end{split}
  \end{equation}
  But, by relation (\ref{id3}), we have
 \begin{equation}\label{E256}
 	\frac{1}{2}D^2\dot{u}^a = \frac{1}{2}D_bD^a\dot{u}^b - \frac{1}{2}(\mbox{curl curl }\dot{u})^a.
 \end{equation}\\[2 mm]
Substituting relation (\ref{E256}) into (\ref{E255}) and using (\ref{E23}) we finally have
\begin{equation}\label{E26}
	\begin{split}
	D_bE^{ab} &=\frac{2}{3}D^a(\mbox{div }\dot{u})+ \frac{2}{3}(\mu - \frac{1}{3}\Theta^2)\dot{u}^a
	- \frac{1}{2}(\mbox{curl curl }\dot{u})^a \\
	 &+ \frac{1}{2}D_b \pi^{ab}.
   \end{split}
\end{equation}

To calculate $(\mbox{curl curl }\dot{u})^a$, we note that the field equation (\ref{E3}) gives
\begin{equation}\label{E261}
\frac{1}{2}(\mbox{curl }\dot{u})^a = \dot{\omega}^{<a>} +\frac{2}{3}\Theta\omega^a.
\end{equation}
Taking $\mbox{curl}$ on both sides and using the commutation relation (\ref{E19}) we have
\begin{equation}\label{E29}
	\frac{1}{2}(\mbox{curl curl }\dot{u})^a = [(\mbox{curl }\omega)^a]\dot{} + \Theta(\mbox{curl }\omega)^a .
\end{equation}

To find the expression for $[(\mbox{curl }\omega)^a]\dot{}$, we use the field equation (\ref{E8}), that gives
\begin{equation}\label{E30}
	(\mbox{curl }\omega)^a = \frac{2}{3}D^a\Theta - q^a\;.
\end{equation}
Taking the dot of the above equation and using relation (\ref{E15}) we get
 \begin{equation}\label{E302}
[(\mbox{curl }\omega)^a]\dot{} = \frac{2}{3}D^a\dot{\Theta} - \frac{1}{9} D^a\Theta^2 -\dot{q}^a\;.
\end{equation}
We now use the field equation (\ref{E1}) and the above gets simplified to
 \begin{equation}\label{E31}
[(\mbox{curl }\omega)^a]\dot{} = \frac{2}{3}D^a(\mbox{div } \dot{u})-\frac{1}{3}D^a\Theta^2-\frac{1}{3}D^a(\mu + 3p)-\dot{q}^a. 
\end{equation} 
Substituting elations (\ref{E30}) and (\ref{E31}) into (\ref{E29}), we obtain
\begin{equation}\label{E32}
	\frac{1}{2}(\mbox{curl curl }\dot{u})^a = \frac{2}{3}D^a(\mbox{div } \dot{u})-\frac{1}{3}D^a(\mu + 3p)-\dot{q}^a-\Theta q^a.
\end{equation}
Using  (\ref{E32}) in (\ref{E26}) we get,
\begin{equation}\label{E33}
	D_bE^{ab} = \frac{2}{3}(\mu - \frac{1}{3}\Theta^2)\dot{u}^a+\frac{1}{3}D^a(\mu + 3p)+\dot{q}^a+\Theta q^a+\frac{1}{2}D_b\pi^{ab}.
\end{equation}
From the field equation (\ref{E11}) we know
\be\label{E331}
D_bE^{ab}+\frac{1}{2}D_b\pi^{ab} -\frac{1}{3}D^a\mu +\frac{1}{3}\Theta q^a=0.
\ee
Using (\ref{E33}) We get
\begin{equation}\label{E34}
 \frac{2}{3}\left(\mu - \frac{1}{3}\Theta^2\right)\dot{u}^a+D^ap+\dot{q}^a+\frac{4}{3}\Theta q^a+D_b\pi^{ab}=0.  
\end{equation}
Again, from field equation (\ref{E6}) we have
\be\label{E35}
\dot{q}^a +\frac{4}{3}\Theta q^a +D^ap +D_b \pi^{ab}= -(\mu +p)\dot{u}^a \,.
\ee
So, the relation (\ref{E34}) becomes
\begin{equation}\label{E351}
	\frac{2}{3}\left(\mu - \frac{1}{3}\Theta^2\right)\dot{u}^a-(\mu +p)\dot{u}^a = 0\,,
\end{equation}
or
\begin{equation}\label{E36}
	\left[\frac{1}{3}\Theta^2 +\frac{1}{2}(\mu + 3p) \right]\dot{u}^a = 0.
\end{equation} 
Now, if the matter obeys strong energy condition, then we have $(\mu + 3p)\ge0$ and of course $\Theta^2\ge0$. Therefore the term in the square bracket of the above equation is strictly non-negative and only vanishes for the trivial case of Minkowski spacetime, that we exclude here. Therefore the matter acceleration must identically vanish for the perturbed spacetime and the geodesics of matter flow lines in the background, map onto themselves in the perturbed manifold. 
\end{proof}

A direct corollary of the above result is
\begin{cor}
If a homogeneous and isotropic spacetime is perturbed in a shear-free way and the matter obeys strong energy condition, then the electric part of the Weyl tensor is half of the anisotropic stress.
\end{cor}
Another straightforward corollary of this theorem is
\begin{cor}
If a homogeneous and isotropic spacetime is perturbed in a shear-free way with the matter obeying strong energy condition and the vorticity is zero at any given instant on a observer's worldline, then it continues to be zero in the entire worldine (as $\omega^a=0\Rightarrow\dot{\omega}^a=0$).
\end{cor}

Thus we showed that if the new constraint emerging due to vanishing of shear has to be spatially consistent with the original constraints, the matter acceleration must identically vanish and the electric part of the Weyl is completely specified by the anisotropic stress. This, however, does not interfere with the definition of magnetic part of Weyl, given by constraint  (\ref{E10}). This can be shown in the following way:
Using relation (\ref{E24}) in (\ref{E10}) we have,
\begin{equation}\label{E372}
	H^{ab}=-\frac{1}{2}D^a\omega^b-\frac{1}{2}D^b\omega^a+\frac{1}{3}h^{ab}({\rm div}\,\omega)\,,
\end{equation}
and,
\begin{equation}\label{E373}
	D_bH^{ab}=-\frac{1}{2}D_bD^a\omega^b-\frac{1}{2}D^2\omega^a+\frac{1}{3}h^{ab}D_bD^a\omega_a\,.
\end{equation}
But by (\ref{E9}) the above becomes,
\begin{equation}\label{E38}
	D_bH^{ab}=-\frac{1}{2}D_bD^a\omega^b-\frac{1}{2}D^2\omega^a. 
\end{equation} 
Now using (\ref{id3}) we know
\begin{equation}\label{E381}
	-\frac{1}{2}D^2\omega^a =  -\frac{1}{2}(\mbox{curl curl }\omega)^a -\frac{1}{2}D_bD^a\omega^b.
\end{equation}
 And (\ref{E38}) becomes
 \begin{equation}\label{E39}
  	D_bH^{ab}= -D_bD^a\omega^b+\frac{1}{2}(\mbox{curl curl }\omega)^a .
 \end{equation}
  But, using (\ref{E23}) and (\ref{E9}),
 \begin{equation}\label{E391}
 	D_bD^a\omega^b= \frac{2}{3}(\mu -\frac{1}{3}\Theta^2 )\omega^a \mbox{,}
 \end{equation}
equation (\ref{E39}) becomes,
\begin{equation}\label{E40}
	D_bH^{ab}= \frac{1}{2}(\mbox{curl curl }\omega)^a - \frac{2}{3}(\mu -\frac{1}{3}\Theta^2 )\omega^a.
\end{equation}

To find the expression for $(\mbox{curl curl }\omega)^a$ we use the field equation (\ref{E8}) that gives
\be\label{E41}
 (\mbox{curl }\omega)^a =  \frac{2}{3}D^a\Theta  - q^a\;.
\ee
Taking the curl on both sides and simplifying using (\ref{E14}) we then have
 \be\label{E42}
(\mbox{curl curl }\omega)^a =  \frac{4}{3}\omega^a \dot{\Theta} - (\mbox{curl }q)^a.
\ee
Substituting (\ref{E42}) into (\ref{E40}) we get
\begin{equation}\label{E43}
	D_bH^{ab}= \frac{2}{3}\omega^a\dot{\Theta} - \frac{1}{2}(\mbox{curl }q)^a  - \frac{2}{3}(\mu -\frac{1}{3}\Theta^2 )\omega^a.
\end{equation}
But from the field equation (\ref{E12}) we have
\begin{equation}\label{E431}
	D_bH^{ab}+(\mu + p)\omega^a + \frac{1}{2}(\mbox{curl }q)^a = 0.
\end{equation}
Plugging in the value of $D_bH^{ab}$ from (\ref{E43})  and simplifying we finally get the constraint
\begin{equation}\label{E44}
	\frac{2}{3}\omega^a(\mbox{div }\dot{u})  = 0.
\end{equation}
We see that the above constraint is identically satisfied if the matter acceleration vanishes.  
\section{Non existence of gravitational waves}    
As discussed in detail in \cite{GosEllis,Roy,Dunsby}, the general way to look for the closed form wave equation of electric or magnetic part of Weyl tensor, is to take the curl of the evolution equation (\ref{E4}) and dot of (\ref{E5}) and then use the commutation relations (\ref{E18}) and (\ref{id4}). However, for this to happen, the equations (\ref{E4}) and (\ref{E5}) must be coupled to each other, just like their electromagnetic Maxwell equation analogue. Interestingly, when the shear tensor is identically zero in the perturbed spacetime, we see by Corollary 1, the magnetic Weyl evolution decouples itself from the electric Weyl evolution and (\ref{E5})
now becomes 
\be\label{E51}
	\dot{H}^{<ab>}= -\Theta H^{ab},
\ee
and on any given worldline the magnetic Weyl has an oscillatory nature given by 
\be\label{E511}
\ddot{H}^{<ab>}= -(\dot{\Theta} H^{ab}+\Theta \dot{H}^{ab})=-[\dot{\Theta} H^{ab}+\Theta (-\Theta H^{ab})].
\ee
Using (\ref{E1}), the above equation becomes
\be\label{E52}
\ddot{H}^{<ab>}= \left[\frac{4}{3}\Theta^2 +\frac{1}{2}\left(\mu + 3p\right)\right] H^{ab}.  
\ee
which is very similar to the vorticity vector. This is hardly surprising, as the magnetic Weyl doesn't contain any tensor modes anymore, and it is purely driven by the vector modes of the vorticity vector. The electric Weyl, however, still contains tensor modes due to the presence of anisotropic stress. To see how these tensor modes behave, we use the usual technique of making all the scalar and the vector modes vanish. That is, we take the vorticity, heat flux and the gradient of all the scalars to be zero. In that case we see that the magnetic part of the Weyl vanishes and we get 
\begin{equation}\label{E531}
	\dot{E}^{<ab>}= -\Theta(E^{ab}+\frac{1}{6}\pi^{ab})-\frac{1}{2}\dot{\pi}^{<ab>},
\end{equation}
and 
\begin{equation}\label{E532}
	\pi^{ab}=2E^{ab}.
\end{equation}
This then implies
\begin{equation}\label{E533} 
\dot{\pi}^{ab}=2\dot{E}^{ab}.
\end{equation}
So (\ref{E531}) becomes
\begin{equation}\label{E534}
	\dot{E}^{<ab>}=-\frac{2}{3}\Theta E^{<ab>}
\end{equation}
\begin{equation}\label{E535}
	\ddot{E}^{<ab>}=-\frac{2}{3}\dot{\Theta} E^{<ab>}-\frac{2}{3}\Theta \dot{E}^{<ab>}.
\end{equation}
Finally using relations (\ref{E1}) and (\ref{E534}) into 
(\ref{E535}) above we can show that the electric part of the Weyl obeys the following tensor oscillatory equation
\begin{equation}\label{E53}
	\ddot{E}^{<ab>}= \frac{1}{3}\left[2 \Theta^2 +(\mu + 3p) \right]E^{ab}.
 \end{equation}
However, we see that this is just homogeneous oscillation of all points at a given time slice and hence there is no gravitational wave propagation. It is interesting to note that, although there is no tensor wave propagation, there exists a vector wave as described in the following proposition:
\begin{prop}
In a general but shear free perturbation of homogeneous and isotropic universes the vorticity vector obeys a wave equation, that is sourced by the curl of the heat flux and is given by
\begin{equation}\label{E54}
\Box \omega^a\equiv	\ddot{\omega}^{<a>}-D^2\omega^a= \frac{1}{3}\omega^a(\mu-3p)-(\mbox{curl }q)^a.
\end{equation}
\end{prop}
\begin{proof} 
By (\ref{E3}) We have, 
\begin{equation}\label{E541}
	\dot{\omega}^{<a>}=-\frac{2}{3}\Theta \omega^a \,,
\end{equation}
and
\begin{equation}\label{E542}
	\ddot{\omega}^{<a>}=\frac{1}{3}\omega^a[2\Theta^2 +(\mu +3p)].
\end{equation}
By (\ref{id3}) we have
\begin{equation}\label{E543}
	D^2\omega^a=D_bD^a\omega^b-(\mbox{curl curl }\omega)^a.
\end{equation}
We know by (\ref{E42}) that
\be\label{E544}
(\mbox{curl curl }\omega)^a =  \frac{4}{3}\omega^a \dot{\Theta} - (\mbox{curl }q)^a\,,
\ee
or, using (\ref{E1}),
\be\label{E545}
(\mbox{curl curl }\omega)^a = - \frac{1}{3}\omega^a\left[\frac{4}{3}\Theta^2+2(\mu + 3p)\right]  - (\mbox{curl }q)^a.
\ee
Substituting (\ref{E545}) into (\ref{E543}) We obtain
\begin{equation}\label{E546}
	\begin{split}
	D^2\omega^a &=D_bD^a\omega^b+\frac{1}{3}\omega^a\left[\frac{4}{3}\Theta^2+2(\mu + 3p)\right]\\
	& + (\mbox{curl }q)^a.
	\end{split}
\end{equation}
Subtracting (\ref{E546}) from (\ref{E542}), we get the result. 
\end{proof}

\section{Discussions}     

In this paper we established the important effects of spacetime shear via the method of negation, by perturbing a homogeneous and isotropic universe in a shear-free fashion. We allowed the matter energy momentum tensor  to have a general perturbation. The key points that emerged from our analysis are as follows:
\begin{enumerate}	   
\item Our analysis once again emphasised the relation of spacetime shear with inhomogeneity and anisotropy that was initially pointed out in \cite{Dadhich}. Inhomogeneity and anisotropy in a spacetime is directly manifested via pressure gradients, anisotropic stresses and heat flux, which in turn generates the matter acceleration. We showed that vanishing of shear necessarily makes the acceleration vanish in the perturbed spacetimes, that makes the matter geodesics in the background to map onto themselves in the perturbed manifold. 
\item We further proved the crucial role of shear in coupling the evolution of electric and magnetic part of the Weyl tensor, that puts gravito-electromagnetism on the same footing as normal electromagnetism. The absence of shear breaks this coupling and furthermore the magnetic part of the Weyl ceases to have any tensor modes (that are purely generated by the shear). Hence the spacetime no longer contains any tensorial waves. It is interesting to note that in general, the shear tensor itself obeys a wave equation in a perturbed FLRW background \cite{Dunsby}. It is very clear that this shear waves mediates the interaction between the electric and magnetic Weyl, to produce the gravitational waves.
\item Although the shear-free perturbation of FLRW universes, do not contain any gravitational waves, it was shown that these contain vorticity waves sourced by the curl of heat flux. If the matter in the perturbed universe continues to be a perfect fluid, then these waves will die down. However if a heat flux with non zero curl is present in the perturbed manifold, then this will act as a forcing term and the vector waves can grow indefinitely resulting in structure formation.
\end{enumerate}
Thus we established the validity of the results obtained in \cite{Matarrese}, even when we allow for the general matter perturbation.

	\begin{acknowledgments}
		RMM and RG are supported by National Research Foundation (NRF), South Africa. SDM 
		acknowledges that this work is based on research supported by the South African Research Chair Initiative of the Department of
		Science and Technology and the National Research Foundation.
	\end{acknowledgments}


\begin{thebibliography}{99} 
	
	\bibitem{Matarrese}  S. Matarrese, O. Pantano and D. Saez, Phys. Rev. Lett. {\bf 72}, 320 (1994).
	
	\bibitem{Godel} K G\"{o}del, {Proceedings of
the International Congress of Mathematicians, Cambridge, Mass.}, \textbf{1}, 175-
181 (Amer. Math. Soc., R.I.), (1952).

\bibitem{Ellis} G. F. R. Ellis , J. Math. Phys. {\bf 8}, 1171, (1967).

\bibitem{Ellisbook} 
		George F. R. Ellis, Roy Maartens, Malcolm A. H. MacCallum 
		{ \it Relativistic Cosmology} (Cambridge University press), Cambridge, (2007).
		 
\bibitem{Nzioki} A. M. Nzioki, R. Goswami, P. K. S. Dunsby and G. F. R. Ellis, Phys. Rev. D{\bf 84}, 124028 (2011). 

\bibitem{GosEllis} R Goswami and G. F. R. Ellis, Class. Quantum Grav. {\bf 38} 085023, (2021). 

\bibitem{Lesame} W. M. Lesame, P. K. S. Dunsby and G. F. R. Ellis, Phys. Rev. D {\bf 52}, 3406, (1995).

\bibitem{Roy1} R. Maartens, Phys. Rev. D {\bf 58}, 124006, (1998).

\bibitem{Bergh} N. van der Bergh and L. Wylleman, Class. Quantum. Grav. {\bf 21} 2291, (2004).

\bibitem{Bergh1} L. Wylleman and N. van der Bergh, Class. Quantum. Grav. {\bf 23} 329, (2006).

\bibitem{Bol} K. Bolejko, Phys. Rev. D {\bf 97}, 083515, (2018).

\bibitem{Bol1} K. Bolejko, Class. Quantum. Grav. {\bf 35} 024003, (2017).
		
\bibitem{SW} J. Stewart,  {\it Advanced General Relativity}, Cambridge University Press, Cambridge (1991).
		 
\bibitem{EllisCovariant} G.~F.~R.~Ellis \& H van Elst,
		{\it Cosmological Models}, Carg\`{e}se Lectures 1998, in Theoretical
		and Observational Cosmology, Ed. M Lachze-Rey, (Dordrecht: Kluwer
		1999), 1. [arXiv:gr-qc/9812046].
		
\bibitem{EllisVarenna}
G. F. R Ellis, Gen. Rel. Gravit. \textbf{41}: 581-660, (2009).

\bibitem{EB} G.~F.~R.~Ellis and M.~Bruni, Phys. Rev. {D} {\bf 40} 1804, (1989).
		
		\bibitem{Roy} R. Maartens and B. Basset, Class. Quantum Grav. {\bf 15 }, 705, (1998).
		\bibitem{Mayala} R. M. Mayala, R. Goswami and S. D. Maharaj, Phys. Rev. D{\bf 103}, 044015, (2021).
		 
		 
		\bibitem{Dadhich} P. S. Joshi, R. Goswami and N. Dadhich, Phys. Rev D{\bf 70}, 087502, (2004)
		
		\bibitem{Dunsby} P. K. S. Dunsby, B. Basset and G. F. R. Ellis, Class. Quant. Grav. {\bf14}, 1215, (1997).
		
		
	\end{thebibliography}
\end{document}